\documentclass[aos]{imsart}

\RequirePackage{amsthm,amsmath,amsfonts,amssymb}
\RequirePackage[numbers,sort&compress]{natbib}
\RequirePackage[colorlinks,citecolor=blue,urlcolor=blue]{hyperref}
\RequirePackage{graphicx}
\usepackage{thmtools} % for table of theorems
\usepackage{tikz} % for tikzpicture
\usepackage{tkz-berge}
\startlocaldefs
\theoremstyle{plain}

\newtheorem{theorem}{Theorem}[section]

\newtheorem{proposition}[theorem]{Proposition}
\newtheorem{corollary}[theorem]{Corollary}
\theoremstyle{definition}

\endlocaldefs

\begin{document}
\newcommand{\bool}[0]{\{0,1\}}
\newcommand{\booln}[0]{\bool^n}
\newcommand{\real}[0]{\mathbb{R}}
\newcommand{\integernonneg}[0]{\mathbb{Z}_{\ge 0}}
\newcommand{\vect}[1]{\boldsymbol{#1}}
\newcommand{\mat}[1]{\mathbf{#1}}
\newcommand{\abs}[1]{\left\vert#1\right\rvert}
\newcommand{\set}[2]{\left\{#1 \, \middle\vert \, #2\right\}}
\newcommand{\indicator}[1]{\left[#1\right]}
\newcommand{\idx}[2]{\left[#1\right]_{#2}}
\newcommand{\powerset}[1]{\mathcal{P}\left(#1\right)}

\newcommand{\prob}[1]{\mathbb{P}\left(#1\right)}
\newcommand{\avg}[1]{\left\langle#1\right\rangle}
\newcommand{\avgpi}[2]{\left\langle#2\right\rangle_{#1}}
\newcommand{\probpi}[2]{\mathbb{P}_{#1}\left(#2\right)}
\newcommand{\normpi}[2]{\lVert#2\rVert_{2,#1}}
\newcommand{\cond}[2]{\left.#1\,\middle\vert\,#2\right.}
\newcommand{\condprob}[2]{\mathbb{P}\left(#1\,\middle\vert\,#2\right)}
\newcommand{\expect}[1]{\mathbb{E}\left[#1\right]}
\newcommand{\expectpi}[2]{\mathbb{E}_{#1}\left[#2\right]}
\newcommand{\condexpect}[2]{\mathbb{E}\left[#1\,\middle\vert\,#2\right]}
\newcommand{\condexpectpi}[3]{\mathbb{E}_{#1}\left[#2\,\middle\vert\,#3\right]}
\newcommand{\var}[1]{\operatorname{Var}\left(#1\right)}
\newcommand{\variancepi}[2]{\operatorname{Var}_{#1}\left(#2\right)}
\newcommand{\condvariancepi}[3]{\operatorname{Var}_{#1}\left(#2\,\middle\vert\,#3\right)}
\newcommand{\varub}[1]{\operatorname{Var}^{\le}\left(#1\right)}
\newcommand{\varlb}[1]{\operatorname{Var}^{\ge}\left(#1\right)}
\newcommand{\varubest}[1]{\widehat{\operatorname{Var}^{\le}}\left(#1\right)}
\newcommand{\varlbest}[1]{\widehat{\operatorname{Var}^{\ge}}\left(#1\right)}
\newcommand{\varubestpi}[2]{\widehat{\operatorname{Var}^{\le}_{#1}}\left(#2\right)}
\newcommand{\covpi}[3]{\operatorname{Cov}_{#1}\left(#2,#3\right)}
\newcommand{\covub}[2]{\operatorname{Cov}^{\le}\left(#1,#2\right)}
\newcommand{\covlb}[2]{\operatorname{Cov}^{\ge}\left(#1,#2\right)}
\newcommand{\covubest}[2]{\widehat{\operatorname{Cov}^{\le}}\left(#1,#2\right)}
\newcommand{\covlbest}[2]{\widehat{\operatorname{Cov}^{\ge}}\left(#1,#2\right)}
\newcommand{\nth}[1]{#1^\text{th}}
\newcommand{\probdist}[2]{\vect{d}_{#1||#2}}
\newcommand{\chisqdiv}[2]{\chi^2\!\left(#1;#2\right)}
\newcommand{\supp}[1]{\operatorname{supp}\left(#1\right)}
\newcommand{\infl}[1]{\operatorname{infl}\left(#1\right)}
\newcommand{\supppi}[2]{\operatorname{supp}_{#1}\left(#2\right)}
\newcommand{\inflpi}[2]{\operatorname{infl}_{#1}\left(#2\right)}
\newcommand{\linspan}{\operatorname{span}}
\newcommand{\bias}[2]{\operatorname{Bias}_{#1}\left(#2\right)}
\newcommand{\hermite}[2]{\operatorname{He}_{#1}\left(#2\right)}

\newcommand{\order}[1]{O\left(#1\right)}
\newcommand{\bigtheta}[1]{\Theta\left(#1\right)}
\newcommand{\littleo}[1]{o\left(#1\right)}
\newcommand{\bigomega}[1]{\Omega\left(#1\right)}
\newcommand{\littleomega}[1]{\omega\left(#1\right)}

\newcommand{\bernoulli}[1]{\operatorname{Bernoulli}\left(#1\right)}

\newcommand{\degree}[1]{\operatorname{deg}\left(#1\right)}

% helpful macros
\makeatletter % sequence of equations
\newcommand{\ceqref}[1]{%
  \textup{(}%
  \@tempswatrue
  \@for\eq:=#1\do{%
    \if@tempswa\@tempswafalse\else, \fi
    \textup{\ref{\eq}}%
  }%
  \textup{)}%
}
\makeatother

% causal estimands
\newcommand{\nodex}[1]{\begin{tikzpicture}[baseline={([yshift=-.6ex]current bounding box.center)}]\SetVertexNoLabel\GraphInit[vstyle=Classic]\SetUpVertex[FillColor=white, MinSize=12pt]\grEmptyPath[prefix=i]{1}\AssignVertexLabel[size=\scriptsize]{i}{$#1$}
\end{tikzpicture}}
\newcommand{\nodet}[1]{\begin{tikzpicture}[baseline={([yshift=-.6ex]current bounding box.center)}]\SetVertexNoLabel\GraphInit[vstyle=Classic]\SetUpVertex[FillColor=green!20, MinSize=12pt]\grEmptyPath[prefix=i]{1}\AssignVertexLabel[size=\scriptsize]{i}{$#1$}
\end{tikzpicture}}
\newcommand{\nodec}[1]{\begin{tikzpicture}[baseline={([yshift=-.6ex]current bounding box.center)}]\SetVertexNoLabel\GraphInit[vstyle=Classic]\SetUpVertex[FillColor=red!20, MinSize=12pt]\grEmptyPath[prefix=i]{1}\AssignVertexLabel[size=\scriptsize]{i}{$#1$}
\end{tikzpicture}}
\newcommand{\nodext}[1]{\begin{tikzpicture}[baseline={([yshift=-.6ex]current bounding box.center)}]\SetVertexNoLabel\GraphInit[vstyle=Classic]\SetUpVertex[FillColor=white, MinSize=12pt]\grPath[prefix=i, RA=1]{2}\AssignVertexLabel[size=\scriptsize]{i}{$#1$}\SetUpVertex[FillColor=green!20, MinSize=12pt]\Vertex[Node]{i1}\end{tikzpicture}}
\newcommand{\nodexc}[1]{\begin{tikzpicture}[baseline={([yshift=-.6ex]current bounding box.center)}]\SetVertexNoLabel\GraphInit[vstyle=Classic]\SetUpVertex[FillColor=white, MinSize=12pt]\grPath[prefix=i, RA=1]{2}\AssignVertexLabel[size=\scriptsize]{i}{$#1$}\SetUpVertex[FillColor=red!20, MinSize=12pt]\Vertex[Node]{i1}\end{tikzpicture}}
\newcommand{\nodett}[1]{\begin{tikzpicture}[baseline={([yshift=-.6ex]current bounding box.center)}]\SetVertexNoLabel\GraphInit[vstyle=Classic]\SetUpVertex[FillColor=green!20, MinSize=12pt]\grPath[prefix=i, RA=1]{2}\AssignVertexLabel[size=\scriptsize]{i}{$#1$}\SetUpVertex[FillColor=green!20, MinSize=12pt]\Vertex[Node]{i1}\end{tikzpicture}}
\newcommand{\nodetc}[1]{\begin{tikzpicture}[baseline={([yshift=-.6ex]current bounding box.center)}]\SetVertexNoLabel\GraphInit[vstyle=Classic]\SetUpVertex[FillColor=green!20, MinSize=12pt]\grPath[prefix=i, RA=1]{2}\AssignVertexLabel[size=\scriptsize]{i}{$#1$}\SetUpVertex[FillColor=red!20, MinSize=12pt]\Vertex[Node]{i1}\end{tikzpicture}}
\newcommand{\nodect}[1]{\begin{tikzpicture}[baseline={([yshift=-.6ex]current bounding box.center)}]\SetVertexNoLabel\GraphInit[vstyle=Classic]\SetUpVertex[FillColor=red!20, MinSize=12pt]\grPath[prefix=i, RA=1]{2}\AssignVertexLabel[size=\scriptsize]{i}{$#1$}\SetUpVertex[FillColor=green!20, MinSize=12pt]\Vertex[Node]{i1}\end{tikzpicture}}
\newcommand{\nodecc}[1]{\begin{tikzpicture}[baseline={([yshift=-.6ex]current bounding box.center)}]\SetVertexNoLabel\GraphInit[vstyle=Classic]\SetUpVertex[FillColor=red!20, MinSize=12pt]\grPath[prefix=i, RA=1]{2}\AssignVertexLabel[size=\scriptsize]{i}{$#1$}\SetUpVertex[FillColor=red!20, MinSize=12pt]\Vertex[Node]{i1}\end{tikzpicture}}
\newcommand{\nodexcc}[1]{\begin{tikzpicture}[baseline={([yshift=-.6ex]current bounding box.center)}]\SetVertexNoLabel\GraphInit[vstyle=Classic]\SetUpVertex[FillColor=red!20, MinSize=12pt]\grEmptyCycle[prefix=i, RA=0.5]{3}\SetUpVertex[FillColor=white, MinSize=12pt]\Vertex[Node]{i2}\Edge(i2)(i1)\Edge(i2)(i0)\AssignVertexLabel[size=\scriptsize]{i}{,,$#1$}\end{tikzpicture}}
\newcommand{\nodexct}[1]{\begin{tikzpicture}[baseline={([yshift=-.6ex]current bounding box.center)}]\SetVertexNoLabel\GraphInit[vstyle=Classic]\SetUpVertex[FillColor=red!20, MinSize=12pt]\grEmptyCycle[prefix=i, RA=0.5]{3}\SetUpVertex[FillColor=white, MinSize=12pt]\Vertex[Node]{i2}\Edge(i2)(i1)\Edge(i2)(i1)\Edge(i2)(i0)\AssignVertexLabel[size=\scriptsize]{i}{,,$#1$}\SetUpVertex[FillColor=green!20, MinSize=12pt]\Vertex[Node]{i0}\end{tikzpicture}}

\newcommand{\eate}[1]{\operatorname{EATE}\left(#1\right)}
\newcommand{\eite}[1]{\operatorname{EAIE}\left(#1\right)}
\newcommand{\eiet}[1]{\operatorname{EAIT}\left(#1\right)}
\newcommand{\eiec}[1]{\operatorname{EAIC}\left(#1\right)}
\newcommand{\eao}[1]{\operatorname{EAO}\left(#1\right)}
\newcommand{\eaoderiv}[2]{\operatorname{EAO}^{(#1)}\left(#2\right)}
\newcommand{\eaoest}[2]{\widehat{\operatorname{EAO}}_{#1}\left(#2\right)}
\begin{frontmatter}
\title{Off-policy Causal Estimation in Networks}
%\title{A sample article title with some additional note\thanksref{t1}}
\runtitle{Off-policy Causal Estimation in Networks}
\thankstext{T1}{We thank Elchanan Mossel for discussion. We are grateful to Christina Lee Yu, Christopher Harshaw, David Choi, and the participants of the CMU Statistics \& Data Science Seminar and 10\textsuperscript{th} Network Science in Economics Conference, for helpful comments. SL acknowledges the support of a Schmidt Science Fellowship.}

\begin{aug}
%%%%%%%%%%%%%%%%%%%%%%%%%%%%%%%%%%%%%%%%%%%%%%%
%% Only one address is permitted per author. %%
%% Only division, organization and e-mail is %%
%% included in the address.                  %%
%% Additional information can be included in %%
%% the Acknowledgments section if necessary. %%
%% ORCID can be inserted by command:         %%
%% \orcid{0000-0000-0000-0000}               %%
%%%%%%%%%%%%%%%%%%%%%%%%%%%%%%%%%%%%%%%%%%%%%%%
\author[A]{\fnms{Sahil}~\snm{Loomba}\ead[label=e1]{s.loomba18@imperial.ac.uk}\orcid{0000-0002-5043-9746}}
\and
\author[B]{\fnms{Dean}~\snm{Eckles}\ead[label=e2]{eckles@mit.edu}\orcid{0000-0001-8439-442X}}
%%%%%%%%%%%%%%%%%%%%%%%%%%%%%%%%%%%%%%%%%%%%%%
%% Addresses                                %%
%%%%%%%%%%%%%%%%%%%%%%%%%%%%%%%%%%%%%%%%%%%%%%
\address[A]{Department of Mathematics, Imperial College London, I-X Centre for AI in Science\printead[presep={ ,\ }]{e1}}

\address[B]{Sloan School of Management, Massachusetts Institute of Technology\printead[presep={,\ }]{e2}}
\end{aug}

\begin{abstract}
In the presence of interference, where the treatment assigned to one unit can affect the outcomes of others, many causal estimands depend on the treatment-assignment policy under which the experiment is conducted. This policy dependence creates a fundamental challenge for off-policy estimation, where the goal is to estimate causal quantities under a hypothetical intervention policy different from the one used to collect data. We study this problem of off-policy estimation of causal effects for heterogeneous Bernoulli policies. By representing exposure-weighted potential outcomes in the biased Fourier basis of the experimental design, we construct, for any prespecified Fourier subspace encoding the assumed interference structure, the unique minimum-$L^2$ weight that transports every function in that subspace. Global and local inverse-probability weights, linear-interference weights, and no-interference weights are special cases. The weight variance is a structured chi-square distance between the experiment and target policies. When the assumed interference structure is misspecified, the introduced bias couples the omitted outcome spectrum with the corresponding policy-shift coefficients, yielding a sharp robustness bound and a bias--variance trade-off. A Fourier-neighborhood-overlap condition gives consistency under structured interference, and we state a Doob-martingale central limit theorem for off-policy estimators. As the variance is not identified, we derive identifiable bounds and associated conservative estimators of the variance. Simulations illustrate these theoretical results for the design and analysis of experiments under network interference and design mismatch.
\end{abstract}



\end{frontmatter}
%%%%%%%%%%%%%%%%%%%%%%%%%%%%%%%%%%%%%%%%%%%%%%
%% Please use \tableofcontents for articles %%
%% with 50 pages and more                   %%
%%%%%%%%%%%%%%%%%%%%%%%%%%%%%%%%%%%%%%%%%%%%%%
%\tableofcontents

\section{Introduction}\label{sec:intro}
Design-based causal inference is commonly developed under the assumption of no interference: one unit's treatment does not affect another unit's outcome \cite{neyman1990causalinference}. In this setting, the average treatment effect (ATE)---measuring the average change in the outcome of a unit when receiving the treatment or not---is often the main causal estimand of interest, and it does not depend on the experimental design. However, the no-interference assumption is often incorrect when units are connected, as in a network, and its relaxation yields a profusion of causal estimands whose true value depends on the design \cite{savje2021ateunknown}. For instance, the expected average treatment effect (EATE) and expected indirect effect (EAIE) capture the effect of flipping the treatment status of a unit on its own and another unit's outcome, averaged over the experimental design; these can be combined to yield an expected average overall effect of treating an additional unit, given the experimental design \cite{hu2022interference}. Moreover, under interference, widely-used summaries of spillovers will generally not identify optimal intervention policies that most improve the expected average outcome (EAO) \cite{loomba2025policyrelevance}. There is thus substantial interest in methods for estimating policy-indexed estimands, like the EAO, that are relevant for policy choice \cite[e.g.,][]{chin2022evaluating,viviano2025policy}.

%\subsection*{Off-policy challenge}
This dependence of the estimand on the policy used to collect the data (i.e. the experimental design) creates an off-policy problem. Say the experimenter collects outcomes under a known \emph{design policy} whereby units are assigned treatment with probabilities given by $\vect{\pi}$, but the analyst wants to estimate causal effects under a new target policy $\vect{\pi}'$. As we show, without restrictions on how treatment assignments affect outcomes, transporting expectations from the design policy to the target policy requires balancing the full assignment distribution: when the design policy has full support, the resulting weight is the full inverse-probability weight. Weighting observed outcomes by it yields an unbiased estimator, but its variance generally grows exponentially with the number of units whose treatments enter the outcome. Hence, useful off-policy evaluation requires imposing structure, which could be a fully specified exposure-response model, or just restrictions on which units' treatments and interaction orders can affect the potential outcomes of a given unit.

%\subsection*{Fourier basis}
Here, we encode that structure in the Fourier spectrum of the outcome function itself. This differs from taking a network and a metric on that network as primitive, as in distance-decay restrictions such as approximate neighborhood interference \cite{leung2022ani}. Fourier support determines which units' treatments matter, while Fourier degree determines the interaction orders through which they matter. Thus, it can represent low effective complexity even in a dense network, and it automatically represents cancellations that a purely network-distance description may obscure. Network restrictions remain valuable as providing conditions for spectral structure, perhaps through some underlying dynamical process of network spillovers, but the Fourier spectrum targets the complexity relevant to policy transport directly.

%\subsection*{Our contributions}
We study this problem for independent Bernoulli designs by representing potential outcomes and exposure indicators as functions on the Boolean cube $\booln$. The corresponding biased Fourier basis makes the policy shift a linear functional of the response spectrum. This yields four contributions. First, for any prespecified Fourier subspace encoding the interference assumptions, we derive the unique minimum-$L^2$ weight that transports expectations from $\vect{\pi}$ to $\vect{\pi}'$ for every function in that subspace. We identify its norm with a structured chi-square policy distance and characterize the exact bias introduced when the assumed subspace omits true Fourier interactions. Second, we show how the variance of this weight depends jointly on the size and order of influence sets, and give a consistency condition that also accounts for overlap between units' influential sets. Third, we establish generic nonidentifiability of randomization variance and derive identifiable covariance bounds with unbiased estimators. Finally, we give a Doob-martingale central limit theorem and simulations are used to illustrate the theoretical results.

%\subsection*{Overview} 
Our use of exposure indicators follows the distinction between using exposure functions to define a causal contrast and to restrict interference, and we use them only for the former \cite{aronow2017interference,savje2023exposure}. Our focus is on transporting such quantities, and expected average outcomes themselves, between policies. Secs. \ref{sec:setup} through \ref{sec:variance} develop the setup, estimator, policy-distance and misspecification analysis, efficiency theory, and conservative variance bounds. Sec. \ref{sec:eao_curve} connects policy curves of the EAO to familiar networked causal contrasts and presents the empirical illustration. Proofs of all results are in Appendix \ref{sec:proofs}, and Appendix \ref{sec:doobclt} develops the Doob-martingale limit argument.

\section{Setup and policy-indexed causal targets}\label{sec:setup}
Consider a finite population of $n$ units. For each unit $i$, the fixed potential-outcome function $y_i:\booln\to\real$ maps the full binary treatment assignment vector $\vect{z}$ to an outcome. The design policy satisfies:
\begin{equation*}
    \vect{Z}\sim\mathbb{P}_{\vect{\pi}},\qquad
    \probpi{\vect{\pi}}{\vect{z}}\triangleq\prod_{i=1}^n\pi_i^{z_i}(1-\pi_i)^{1-z_i},
    \qquad \vect{\pi}\in(0,1)^n,
\end{equation*}
and all randomness is induced by this design. Throughout, lower-case letters denote fixed functions of an assignment and the corresponding upper-case letters denote their realized random variables. Thus $F\triangleq f(\vect{Z})$ for a generic function $f$, and the observed outcome is $Y_i\triangleq y_i(\vect{Z})$. The results below are first stated for a generic Boolean function $f:\booln\to\real$. We bring them back to causal inference by taking $f$ to be a potential-outcome function or an exposure-weighted potential-outcome function, and then averaging the resulting unit-level quantities. This keeps the Boolean-function arguments general while making their implications for EAO and causal contrasts explicit.

Let $t_i^+,t_i^-:\booln\to\bool$ be mutually exclusive exposure indicator functions, that define the positive and negative conditions being contrasted, i.e. they do not by themselves restrict how $y_i$ depends on treatment \cite{savje2023exposure}. For any target product policy $\vect{\pi}'\in[0,1]^n$, define:
\begin{equation*}
    p_i^+(\vect{\pi}')\triangleq\expectpi{\vect{\pi}'}{T_i^+},\qquad p_i^-(\vect{\pi}')\triangleq\expectpi{\vect{\pi}'}{T_i^-}.
\end{equation*}
Whenever both probabilities are positive, the policy-indexed contrast is defined as:
\begin{equation}\label{eq:policy_contrast}
    \delta_{\vect{\pi}'}\triangleq\frac{1}{n}\sum_{i=1}^n\left(\frac{\expectpi{\vect{\pi}'}{Y_iT_i^+}}{p_i^+(\vect{\pi}')}-\frac{\expectpi{\vect{\pi}'}{Y_iT_i^-}}{p_i^-(\vect{\pi}')}\right).
\end{equation}
This includes expected direct and exposure contrasts. The expected average outcome is:
\begin{equation*}
    \eao{\vect{\pi}'}
    \triangleq\frac{1}{n}\sum_{i=1}^n\expectpi{\vect{\pi}'}{Y_i}
\end{equation*}
is a one-term target and does not require a second exposure indicator. At the experimental policy, the Horvitz--Thompson estimator \cite{horvitz1952sampling}:
\begin{equation*}
    \widehat{\delta}_{\vect{\pi}}\triangleq\frac{1}{n}\sum_{i=1}^nY_i\left(\frac{T_i^+}{p_i^+(\vect{\pi})}-\frac{T_i^-}{p_i^-(\vect{\pi})}\right)
\end{equation*}
is unbiased for $\delta_{\vect{\pi}}$. Our objective is to estimate $\delta_{\vect{\pi}'}$ when $\vect{\pi}'\ne\vect{\pi}$ from the same realized experiment $\vect{Z}$.

\section{Spectral transport across treatment policies}\label{sec:spectral_transport}
For a function $f:\booln\to\real$, let
\begin{equation*}
    \chi_{\vect{\pi}}^V(\vect{z})\triangleq\prod_{j\in V}\frac{z_j-\pi_j}{\sqrt{\pi_j(1-\pi_j)}},\qquad V\subseteq[n],
\end{equation*}
be $\vect{\pi}$-biased Fourier characters with $\chi_{\vect{\pi}}^\emptyset\triangleq1$. These functions form an orthonormal basis under the design policy $\mathbb{P}_{\vect{\pi}}$, i.e. $\expectpi{\vect{\pi}}{\chi_{\vect{\pi}}^V(\vect{Z})\chi_{\vect{\pi}}^W(\vect{Z})}=1$ if $V=W$ and $0$ otherwise, yielding
\begin{equation}\label{eq:boolfourier}
    f(\vect{z})=\sum_{V\subseteq[n]}\hat{f}_{\vect{\pi}}(V)\,\chi_{\vect{\pi}}^V(\vect{z})
\end{equation}
as the corresponding product-distribution analogue of the usual Fourier expansion of Boolean functions \cite{odonnell2014boolean}. We write $\supppi{\vect{\pi}}{f}\triangleq\set{V\subseteq[n]}{\hat{f}_{\vect{\pi}}(V)\ne 0}$ for the Fourier support of $f$ under a policy $\vect{\pi}$. For any square-integrable function $g$, define its $L^2$ norm under that policy:
\begin{equation*}
    \normpi{\vect{\pi}}{g}\triangleq\expectpi{\vect{\pi}}{G^2}^{1/2}.
\end{equation*}

Define the policy-shift coefficients:
\begin{equation*}
    \Delta_{\vect{\pi}\vect{\pi}'}^j\triangleq\frac{\pi_j'-\pi_j}{\sqrt{\pi_j(1-\pi_j)}},\qquad\Delta_{\vect{\pi}\vect{\pi}'}^V\triangleq\prod_{j\in V}\Delta_{\vect{\pi}\vect{\pi}'}^j,
\end{equation*}
with $\Delta_{\vect{\pi}\vect{\pi}'}^\emptyset\triangleq 1$. The two policy subscripts denote the design and target policies. Independence under the target product policy gives:
\begin{equation}
    \label{eq:target_character_mean}
    \expectpi{\vect{\pi}'}{\chi_{\vect{\pi}}^V(\vect{Z})}=\Delta_{\vect{\pi}\vect{\pi}'}^V.
\end{equation}
In other words, $\Delta_{\vect{\pi}\vect{\pi}'}^V$ is the expectation of each design-policy basis function under the target policy. From \eqref{eq:boolfourier} we immediately see how this could be applied to policy transport:
\begin{equation*}
    \expectpi{\vect{\pi}'}{F}=\sum_{V\subseteq[n]}\hat{f}_{\vect{\pi}}(V)\Delta_{\vect{\pi}\vect{\pi}'}^V.
\end{equation*}
That is, the vector $\set{\Delta_{\vect{\pi}\vect{\pi}'}^V}{V\subseteq[n]}$ gives the coefficients of the target-expectation functional in the design-policy basis, and can be used for policy transport.

In our formalism, an interference assumption or mechanism restriction specifies a support set $\mathcal{S}\subseteq\powerset{[n]}$ containing $\emptyset$, where $\powerset{\cdot}$ denotes the power set. Define a support restriction's associated Fourier subspace:
\begin{equation*}
    \mathcal{H}_{\mathcal{S}}\triangleq\linspan\set{\chi_{\vect{\pi}}^V}{V\in\mathcal{S}}=\set{f}{\supppi{\vect{\pi}}{f}\subseteq\mathcal{S}}.
\end{equation*}
The definition of $\mathcal{H}_{\mathcal{S}}$ is descriptive rather than an assumption, since $\mathcal{S}$ simply lists the Fourier components that the analyst allows for the policy transport. Define an $\mathcal{S}$-restricted weight whose $\vect{\pi}$-biased Fourier coefficients are exactly the corresponding policy-shift coefficients:
\begin{equation}
    \label{eq:spectral_transport_weight}
    w_{\vect{\pi}\vect{\pi}'}^{\mathcal{S}}(\vect{z})\triangleq\sum_{V\in\mathcal{S}}\Delta_{\vect{\pi}\vect{\pi}'}^V\,\chi_{\vect{\pi}}^V(\vect{z}).
\end{equation}

\begin{proposition}[Spectral transport]\label{prop:spectral_transport}
If $f\in\mathcal{H}_{\mathcal{S}}$, then:
\begin{equation*}
    \expectpi{\vect{\pi}}{F W_{\vect{\pi}\vect{\pi}'}^{\mathcal{S}}}=\expectpi{\vect{\pi}'}{F}.
\end{equation*}
Among all square-integrable weighting functions $u$ for which $\expectpi{\vect{\pi}}{FU}=\expectpi{\vect{\pi}'}{F}$ for every $f\in\mathcal{H}_{\mathcal{S}}$, $w_{\vect{\pi}\vect{\pi}'}^{\mathcal{S}}$ uniquely minimizes $\normpi{\vect{\pi}}{u}$. The variance of its realization is:
\begin{equation}
    \label{eq:spectral_weight_variance}
    \variancepi{\vect{\pi}}{W_{\vect{\pi}\vect{\pi}'}^{\mathcal{S}}}=\sum_{V\in\mathcal{S}\setminus\{\emptyset\}}\left(\Delta_{\vect{\pi}\vect{\pi}'}^V\right)^2.
\end{equation}
\end{proposition}
\noindent\textbf{Interpretation.} For every function whose Fourier support lies in $\mathcal S$, weighting under the design policy by $W_{\vect{\pi}\vect{\pi}'}^{\mathcal S}$ exactly yields its target-policy mean. Among all weights that achieve this simultaneously for every such function, $W_{\vect{\pi}\vect{\pi}'}^{\mathcal S}$ has the smallest $L^2(\mathbb{P}_{\vect{\pi}})$ norm and hence also has the smallest variance under the design policy. Thus, the weight precisely captures the retained Fourier directions without spending variance on directions excluded from $\mathcal{S}$.

When $\mathcal{S}=\powerset{\Gamma}$ for a set of units $\Gamma\subseteq[n]$, the sum in \eqref{eq:spectral_transport_weight} factorizes as an inverse-probability weight:
\begin{equation*}
    w_{\vect{\pi}\vect{\pi}'}^{\powerset{\Gamma}}(\vect{z})=\prod_{i\in\Gamma} \left(\frac{\pi_i'}{\pi_i}\right)^{z_i}\left(\frac{1-\pi_i'}{1-\pi_i}\right)^{1-z_i}=\frac{\probpi{\vect{\pi}'}{\vect{z};\Gamma}}{\probpi{\vect{\pi}}{\vect{z};\Gamma}},
\end{equation*}
where we use $\probpi{\vect{\pi}}{\cdot;\Gamma}$ for the probability of a treatment assignment restricted to the units in $\Gamma$, with $\probpi{\vect{\pi}}{\cdot;\emptyset}\triangleq 1$. As a special case, consider the \emph{global} inverse-probability weight, arising from considering the probability of the entire treatment assignment vector $\vect{z}$ under the design policy. This arises in this construction when $\Gamma=[n]$, which gives a useful change-of-measure interpretation to \eqref{eq:target_character_mean}: $\Delta_{\vect{\pi}\vect{\pi}'}^V=\hat{w}_{\vect{\pi}\vect{\pi}',\vect{\pi}}^{\powerset{[n]}}(V).$ That is, the expectation of each design-policy basis function under the target policy is exactly the corresponding Fourier coefficient of the full inverse-probability weight. Since an $\mathcal{S}$-restricted weight picks out a subset of these Fourier coefficients, it is best seen as an orthogonal projection of the full inverse-probability weight onto the subspace $\mathcal{H}_\mathcal{S}$. More formally, for any $g$, define its orthogonal projection onto $\mathcal{H}_\mathcal{S}$ by:
\begin{equation}
    \label{eq:fourier_projection}
    P_{\mathcal{S}}g(\vect{z})\triangleq\sum_{V\in\mathcal{S}}\hat{g}_{\vect{\pi}}(V)\,\chi_{\vect{\pi}}^V(\vect{z}).
\end{equation}
Then:
\begin{equation*}
    w_{\vect{\pi}\vect{\pi}'}^{\mathcal{S}}=P_{\mathcal{S}}w_{\vect{\pi}\vect{\pi}'}^{\powerset{[n]}}.
\end{equation*}
When $\mathcal{S}=\powerset{[n]}$, the balance conditions apply to every source-policy Fourier character. Because these characters form a basis for all functions on $\booln$, satisfying these conditions transports the expectation of every function from $\vect{\pi}$ to $\vect{\pi}'$.

We remark that for a general support set $\mathcal{S}$, $W_{\vect{\pi}\vect{\pi}'}^{\mathcal{S}}$ can be negative. This does not give it an interpretation as a probability ratio, rather it is an orthogonal projection of a nonnegative probability ratio, and projection onto a linear subspace need not preserve positivity. Its interpretation is instead as a signed balancing weight, since:
\begin{equation}
    \label{eq:spectral_balance_constraints}
    \expectpi{\vect{\pi}}{W_{\vect{\pi}\vect{\pi}'}^{\mathcal{S}}\chi_{\vect{\pi}}^V(\vect{Z})}=\expectpi{\vect{\pi}'}{\chi_{\vect{\pi}}^V(\vect{Z})},
\end{equation}
i.e. they encode the correction required to balance the retained Fourier coefficients after the remaining coefficients have been discarded.

We return now from the transport of general Boolean functions to that of exposure-conditional functions: $x_i^+(\vect{z})\triangleq y_i(\vect{z})\,t_i^+(\vect{z})$ and $x_i^-(\vect{z})\triangleq y_i(\vect{z})\,t_i^-(\vect{z})$. Suppose that sets $\mathcal{S}_i^+$ and $\mathcal{S}_i^-$ contain their respective Fourier supports and that $ p_i^+(\vect{\pi}'), p_i^-(\vect{\pi}')>0$. Then:
\begin{equation}
    \label{eq:offpolicy_estimator}
    \widehat{\delta}_{\vect{\pi}'}=\frac{1}{n}\sum_{i=1}^nY_i\left(\frac{T_i^+W_{\vect{\pi}\vect{\pi}'}^{\mathcal{S}_i^+}}{p_i^+(\vect{\pi}')}-\frac{T_i^-W_{\vect{\pi}\vect{\pi}'}^{\mathcal{S}_i^-}}{p_i^-(\vect{\pi}')}\right)
\end{equation}
is unbiased for $\delta_{\vect{\pi}'}$. The EAO is the one-term special case with $t_i^+\triangleq 1$ and $t_i^-$ omitted.

\section{Policy distance and misspecified interference}\label{sec:misspecification}

The coefficients $\Delta_{\vect{\pi}\vect{\pi}'}^j$ quantify the unit-level shift from the design policy to the target policy. They can be seen as encoding a divergence between the two policies:
\begin{equation*}
    \left(\Delta_{\vect{\pi}\vect{\pi}'}^i\right)^2=\frac{(\pi_i'-\pi_i)^2}{\pi_i(1-\pi_i)}=\chisqdiv{\bernoulli{\pi_i'}}{\bernoulli{\pi_i}},
\end{equation*}
where the chi-squared divergence is defined by $\chisqdiv{Q}{P}\triangleq\expectpi{P}{\left(\frac{dQ}{dP}-1\right)^2}$. For product policies,
\begin{equation*}
    \chisqdiv{\mathbb{P}_{\vect{\pi}'}}{\mathbb{P}_{\vect{\pi}}}=\prod_{i=1}^n\left(1+\left(\Delta_{\vect{\pi}\vect{\pi}'}^i\right)^2\right)-1.
\end{equation*}
For any Fourier index set $\mathcal{I}\subseteq\powerset{[n]}$, define:
\begin{equation*}
    d_{\mathcal{I}}^2(\vect{\pi}',\vect{\pi})\triangleq\sum_{V\in\mathcal{I}\setminus\{\emptyset\}}\left(\Delta_{\vect{\pi}\vect{\pi}'}^V\right)^2
\end{equation*}
as the \emph{structured} chi-squared policy distance along the Fourier directions in $\mathcal{I}$. In particular,
\begin{equation*}
    d_{\mathcal{S}}^2(\vect{\pi}',\vect{\pi})=\normpi{\vect{\pi}}{w_{\vect{\pi}\vect{\pi}'}^{\mathcal{S}}-1}^2=\variancepi{\vect{\pi}}{W_{\vect{\pi}\vect{\pi}'}^{\mathcal{S}}}.
\end{equation*}
It encodes the squared $L^2(\mathbb{P}_{\vect{\pi}})$ norm of the non-constant part of the projected inverse-probability weight. Only policy shifts along Fourier interactions retained in $\mathcal{S}$ contribute to it. We note that, as with chi-squared divergence, the word ``distance'' is descriptive here, and not used in the sense of a metric. Setting $\mathcal{I}=\mathcal{S}^c$, where complements are taken relative to $\powerset{[n]}$, measures the shift along the omitted directions, so the full product-policy discrepancy splits orthogonally as:
\begin{equation*}
    \chisqdiv{\mathbb{P}_{\vect{\pi}'}}{\mathbb{P}_{\vect{\pi}}}=d_{\mathcal{S}}^2(\vect{\pi}',\vect{\pi})+d_{\mathcal{S}^c}^2(\vect{\pi}',\vect{\pi}).
\end{equation*}
The two terms correspond to policy-shift directions retained and omitted by the assumed Fourier subspace. The same coefficients also determine the sensitivity to a misspecified interference assumption. Recall from \eqref{eq:fourier_projection} that $P_{\mathcal{S}}f$ retains precisely the Fourier terms indexed by $\mathcal{S}$, and therefore $f-P_{\mathcal{S}}f$ is the part of the true function omitted by the analyst.

\begin{proposition}[Spectral misspecification bias]
\label{prop:misspecification_bias}
For any square-integrable $f$,
\begin{equation}
    \begin{split}
    \bias{\mathcal{S}}{f;\vect{\pi},\vect{\pi}'}&\triangleq\expectpi{\vect{\pi}}{FW_{\vect{\pi}\vect{\pi}'}^{\mathcal{S}}}-\expectpi{\vect{\pi}'}{F}\\
    &=-\sum_{V\notin\mathcal{S}}\hat{f}_{\vect{\pi}}(V)\Delta_{\vect{\pi}\vect{\pi}'}^V.
    \end{split}
    \label{eq:exact_misspecification_bias}
\end{equation}
Moreover,
\begin{equation}
    \abs{\bias{\mathcal{S}}{f;\vect{\pi},\vect{\pi}'}}\le\normpi{\vect{\pi}}{f-P_{\mathcal{S}}f}\cdot d_{\mathcal{S}^c}(\vect{\pi}',\vect{\pi}).
    \label{eq:misspecification_bias_bound}
\end{equation}
Consequently, for every $c>0$, the largest possible absolute bias among functions satisfying $\normpi{\vect{\pi}}{f-P_{\mathcal{S}}f}\le c$ is exactly $c\cdot d_{\mathcal{S}^c}(\vect{\pi}',\vect{\pi})$.
\end{proposition}
\noindent\textbf{Interpretation.} The bias is the inner product of two omitted-spectrum vectors: the Fourier coefficients of the true response and the design-to-target policy shifts. Support containment is therefore sufficient but not necessary for unbiasedness: omitted terms can cancel, and any interaction containing only units with $\pi_i'=\pi_i$ makes no contribution. At the on-policy point all nonempty $\Delta_{\vect{\pi}\vect{\pi}'}^V$ vanish. Under a homogeneous shift from $p$ to $q$, define $\epsilon_{pq}\triangleq\frac{\abs{q-p}}{\sqrt{p(1-p)}}$. An omitted interaction of order $r$ is then multiplied in absolute value by $\epsilon_{pq}^r$. Define:
\begin{align*}
    m_r&\triangleq\abs{\set{V}{V\notin\mathcal{S},\abs{V}=r}},\\
    e_r(f)&\triangleq\left(\sum_{\substack{V\notin\mathcal{S}\\\abs{V}=r}}\hat{f}_{\vect{\pi}}(V)^2\right)^{\frac{1}{2}}.
\end{align*}
If the first omitted order is $r_0$, then:
\begin{equation*}
    \abs{\bias{\mathcal{S}}{f;p,q}}\le\sum_{r=r_0}^n\epsilon_{pq}^r\sqrt{m_r}\,e_r(f).
\end{equation*}
Thus, misspecification beginning at order $r_0$ produces $\order{\abs{q-p}^{r_0}}$ local bias as $q\to p$. Higher-order misspecification can, in principle, be tolerated for nearby policies when its Fourier energy $e_r$ and multiplicity $m_r$ are controlled. Adding a set $V$ to $\mathcal{S}$ increases the weight variance by $\left(\Delta_{\vect{\pi}\vect{\pi}'}^V\right)^2$ and removes the corresponding bias term $-\hat{f}_{\vect{\pi}}(V)\Delta_{\vect{\pi}\vect{\pi}'}^V$.

For the causal contrast, the relevant true functions are $x_i^+\triangleq y_it_i^+$ and $x_i^-\triangleq y_it_i^-$, rather than $y_i$ alone (except for EAO, where $t_i^+\triangleq 1$). The exact bias of \eqref{eq:offpolicy_estimator} is
\begin{equation*}
    -\frac{1}{n}\sum_{i=1}^n\left(\frac{\sum_{V\notin\mathcal{S}_i^+}\hat{x}^+_{i,\vect{\pi}}(V)\Delta_{\vect{\pi}\vect{\pi}'}^V}{p_i^+(\vect{\pi}')}-\frac{\sum_{V\notin\mathcal{S}_i^-}\hat{x}^-_{i,\vect{\pi}}(V)\Delta_{\vect{\pi}\vect{\pi}'}^V}{p_i^-(\vect{\pi}')}\right).
\end{equation*}
This result separates the contribution to the bias of mechanism misspecification from the distance over which the analyst attempts to policy-transport a contrast.

\section{Efficiency under structured interference}\label{sec:efficiency}

With no restriction on $f$, Proposition \ref{prop:spectral_transport} uses the global inverse-probability weight. Its variance is $\prod_{i=1}^n(1+(\Delta_{\vect{\pi}\vect{\pi}'}^i)^2)-1$, which is exponential in $n$ when the policy shift is non-vanishing. Structural restrictions reduce the Fourier subspace and hence remove unnecessary components of the weight. For a function $f$, define its degree as:
\begin{equation*}
    \degree{f}\triangleq\max\set{\abs{V}}{\hat{f}_{\vect{\pi}}(V)\ne0},
\end{equation*}
and its influence set as:
\begin{equation*}
    \Gamma_f\triangleq\bigcup_{V\in\supppi{\vect{\pi}}{f}}V.
\end{equation*}
For products, $\degree{fg}\le\degree{f}+\degree{g}$ and $\Gamma_{fg}\subseteq\Gamma_f\cup\Gamma_g$. (The inclusions may be strict because terms can cancel on the Boolean cube.) Thus, restrictions on \emph{both} the potential outcome and exposure functions restrict the support of the exposure-conditional functions $x_i^+\triangleq y_it_i^+$ and $x_i^-\triangleq y_it_i^-$. Table \ref{tab:interf_weights} gives the weights needed for the four Fourier support restrictions used in the literature.

\begin{table}[t]
\centering
\begin{tabular}{@{}lll@{}}
\hline
Assumption & Transport weight & Weight variance under $\vect{\pi}$\\
\hline
Complete on $\Gamma$
& $\prod_{i\in\Gamma}\left(1+\Delta^i\chi_{\vect{\pi}}^i(\vect{Z})\right)$
& $\prod_{i\in\Gamma}\left(1+\left(\Delta^i\right)^2\right)-1$\\[5pt]
Degree at most $d$
& $\sum_{\substack{V\subseteq\Gamma\\|V|\le d}}\Delta^V\chi_{\vect{\pi}}^V(\vect{Z})$
& $\sum_{\substack{V\subseteq\Gamma\\1\le|V|\le d}}\left(\Delta^V\right)^2$\\[7pt]
Linear
& $1+\sum_{i\in\Gamma}\Delta^i\chi_{\vect{\pi}}^i(\vect{Z})$
& $\sum_{i\in\Gamma}\left(\Delta^i\right)^2$\\[5pt]
SUTVA for unit $i$
& $1+\Delta^i\chi_{\vect{\pi}}^i(\vect{Z})$
& $(\Delta^i)^2$\\
\hline
\end{tabular}
\caption{\textbf{Minimum-variance spectral transport weights for common interference assumptions.} \\Minimum variance is among weights that transport every function in the stated subspace. For readability, we have suppressed the dependence on design and target policies i.e. $\Delta^i=\Delta_{\vect{\pi}\vect{\pi}'}^i$ and $\Delta^V=\Delta_{\vect{\pi}\vect{\pi}'}^V$.}
\label{tab:interf_weights}
\end{table}

We next give a simple consistency condition that separates the size of each unit's Fourier neighborhood from the amount of overlap between neighborhoods. Consider a sequence of finite populations indexed by $n$, with population $[n]$ assigned under the product policy $\vect{\pi}^{(n)}$. Let:
\begin{equation*}
    \widehat{\theta}^{(n)}\triangleq\frac{1}{n}\sum_{i=1}^n\widehat{\theta}_i^{(n)}
\end{equation*}
be a transported mean or causal estimator, possibly based on a misspecified Fourier subspace, where the unit-level estimator $\widehat{\theta}_i^{(n)}$ depends only on the treatment assignments of units in $\Gamma_i^{(n)}$. Define the spectral ``overlap'' degree:
\begin{equation*}
    D^{(n)}\triangleq\max_{i\in[n]}\abs{\set{j\in[n]\setminus\{i\}}{\Gamma_i^{(n)}\cap\Gamma_j^{(n)}\ne\emptyset}}.
\end{equation*}

\begin{theorem}[Consistency under spectral overlap]\label{prop:consistency}
If $\max_{i\in[n]}\expectpi{\vect{\pi}^{(n)}}{\left(\widehat{\theta}_i^{(n)}\right)^2}\le B^{(n)}$, then:
\begin{equation}
    \variancepi{\vect{\pi}^{(n)}}{\widehat{\theta}^{(n)}}\le\frac{B^{(n)}\left(D^{(n)}+1\right)}{n}.
    \label{eq:consistency_variance_bound}
\end{equation}
Consequently, if $B^{(n)}\left(D^{(n)}+1\right)=\littleo{n}$, then $\widehat{\theta}^{(n)}\xrightarrow{p}\expectpi{\vect{\pi}^{(n)}}{\widehat{\theta}^{(n)}}$.
\end{theorem}

\noindent\textbf{Interpretation.} Overlap quantifies a dependence penalty: relative to independent summands, the variance bound expands by a factor of at most $D^{(n)}+1$. This motivates defining an overlap-adjusted effective sample size:
\begin{equation*}
    n_{\mathrm{eff}}^{(n)}\triangleq\frac{n}{D^{(n)}+1}.
\end{equation*}
If $D^{(n)}=\littleo{1}$, which implies eventual disjointness of overlap, or $D^{(n)}=\order{1}$, then $n_{\mathrm{eff}}^{(n)}=\bigtheta{n}$. If $D^{(n)}=\littleomega{1}$ but $D^{(n)}=\littleo{n}$, the effective size still diverges. When $D^{(n)}=\bigtheta{n}$, this overlap bound does not vanish from bounded second moments alone, and additional cancellation would be required to establish consistency. We emphasize that the sets $\Gamma_i^{(n)}$ need not be network neighborhoods, rather they may be any sets of units whose treatment assignments suffice to influence the summands $\widehat{\theta}_i$.

For a fixed homogeneous shift from $p$ to $q$, write $c\triangleq\frac{(q-p)^2}{p(1-p)}$ and $s^{(n)}\triangleq\max_i\abs{\Gamma_i^{(n)}}$. If the remaining factors in $\widehat{\theta}_i^{(n)}$ are uniformly bounded, complete local interference gives $B^{(n)}=\order{(1+c)^{s^{(n)}}}$, whereas linear interference gives $B^{(n)}=\order{1+cs^{(n)}}$. Therefore, Theorem \ref{prop:consistency} gives sufficient conditions on the size of influence sets for consistency:
\begin{equation*}
    \begin{aligned}
    (1+c)^{s^{(n)}}&=\littleo{n_{\mathrm{eff}}^{(n)}} && \text{under complete local interference,}\\
    s^{(n)}&=\littleo{n_{\mathrm{eff}}^{(n)}} && \text{under linear interference.}
    \end{aligned}
\end{equation*}
Here $s^{(n)}$ is the largest number of units whose treatments can influence a given unit's summand, while $D^{(n)}$ is the largest number of other summands sharing at least one such unit that determines $n_{\mathrm{eff}}^{(n)}$. Under complete local interference, logarithmically growing influence sets are permitted, i.e. it is sufficient that $s^{(n)}\le\kappa\log n_{\mathrm{eff}}^{(n)}$ for some fixed $\kappa<\log(1+c)^{-1}$. When $D^{(n)}=\order{1}$, we have $n_{\mathrm{eff}}^{(n)}=\bigtheta{n}$, so the same condition holds with $\log n$ instead of $\log n_{\mathrm{eff}}^{(n)}$. By contrast, linear interference permits the substantially larger influence regime $s^{(n)}=\littleo{n_{\mathrm{eff}}^{(n)}}$. 

Two common network-induced influence structures make these rates concrete. For disjoint fully-connected networks of size at most $m_n$, we have $s^{(n)}=\order{m_n}$ and $D^{(n)}=\order{m_n}$, and the sufficient conditions become $m_n(1+c)^{m_n}=\littleo{n}$ under complete interference within clusters and $m_n^2=\littleo{n}$ under linear interference. For a fixed $\ell>0$, suppose that each unit depends on the treatment of its own unit and those of its $\ell$-hop neighbors in a network with maximum degree $k_n$. Then we have $s^{(n)}=\order{k_n^\ell}$ and $D^{(n)}=\order{k_n^{2\ell}}$, because only units within $2\ell$-hops can have overlap. The corresponding sufficient conditions are $k_n^{2\ell}(1+c)^{k_n^\ell}=\littleo{n}$ and $k_n^{3\ell}=\littleo{n}$, respectively. These are sufficient worst-case scalings, and a finer dependence argument could improve them for particular networks.

For the EAO, the preceding bias formula and the overlap-based variance bound combine directly. This makes the trade-off in choosing the Fourier subspace explicit. In what follows, we fix $n$ in the preceding sequence and suppress the dependence on it for readability. 

\begin{corollary}[Spectral bias--variance envelope]
\label{prop:spectral_bias_variance}
Let $\vect{\pi}'$ be the target policy, let $\mathcal{S}_i$ be unit $i$'s retained Fourier support, and define:
\begin{align*}
    \widehat{\mu}_{\mathcal{S}}\triangleq\frac{1}{n}\sum_{i=1}^nY_iW^{\mathcal{S}_i},\qquad\mu_{\vect{\pi}'}\triangleq\frac{1}{n}\sum_{i=1}^n\expectpi{\vect{\pi}'}{Y_i},\qquad r_i \triangleq\normpi{\vect{\pi}}{y_i-P_{\mathcal{S}_i}y_i}.
\end{align*}
Let $\abs{y_i(\vect{z})}\le M$ and $D$ be the overlap degree of the influence sets on which these summands depend. Then:
\begin{align*}
    \expectpi{\vect{\pi}}{(\widehat\mu_{\mathcal{S}}-\mu_{\vect{\pi}'})^2} 
    &=\variancepi{\vect{\pi}}{\widehat\mu_{\mathcal{S}}}+\left(\frac{1}{n}\sum_{i=1}^n\sum_{V\notin\mathcal{S}_i}\hat{y}_{i,\vect{\pi}}(V)\Delta_{\vect{\pi}\vect{\pi}'}^V\right)^2 \\
    &\le\frac{M^2(D+1)}{n}\max_{i\in[n]}(1+d_{\mathcal{S}_i}^2(\vect{\pi}',\vect{\pi}))+\left(\frac{1}{n}\sum_{i=1}^nr_id_{\mathcal{S}_i^c}(\vect{\pi}',\vect{\pi})\right)^2.
\end{align*}
\end{corollary}

\noindent\textbf{Interpretation.} The first term is an overlap-adjusted variance cost determined by the retained spectrum, while the second bounds the squared bias contributed by the omitted spectrum. Enlarging $\mathcal{S}_i$ can reduce the second term only by potentially increasing the first, so support choice is better viewed as an explicit bias--variance decision rather than a binary correct-or-incorrect model specification. 

For causal contrasts, the same decomposition applies to the exposure-weighted functions $x_i^+\triangleq y_it_i^+$ and $x_i^-\triangleq y_it_i^-$, with their target exposure probabilities carried through as in \eqref{eq:offpolicy_estimator}.

\section{Variance nonidentifiability and conservative bounds}\label{sec:variance}

Even when the estimator is unbiased, its randomization variance is generally not identifiable from a single realized assignment without additional restrictions. This is a version of the familiar non-identifiability of the Neymanian variance of the difference-in-means in the absence of interference \cite{neyman1990causalinference}, due to never jointly observing multiple potential outcomes per unit. Here it applies to arbitrary functions of a treatment vector.

\begin{proposition}[Generic variance nonidentifiability]\label{prop:variance_nonidentification}
Consider a policy $\vect{\pi}$ that assigns positive probability to at least two treatment vectors, and consider function $f$ that is otherwise unrestricted. There is no statistic depending only on the observed pair $(\vect{Z},F)$ that is unbiased for $\variancepi{\vect{\pi}}{F}$ for every function $f$.
\end{proposition}

\noindent\textbf{Interpretation.} The obstruction is fundamentally observational, since one realized assignment reveals one value $f(\vect{Z})$, while the variance contains products of potential outcomes under mutually exclusive assignments. Interference restrictions nevertheless allow identifiable bounds. 

\begin{proposition}[Identifiable covariance bounds]\label{prop:covariance_bounds}
Covariance between $f,g$ satisfies:
\begin{align}
    -\frac{1}{2}\expectpi{\vect{\pi}}{(F-G)^2[1-\probpi{\vect{\pi}}{\vect{Z};\Gamma_f\cap\Gamma_g}]}\le\,&\covpi{\vect{\pi}}{F}{G}\nonumber\\
    \le\,&\frac{1}{2}\expectpi{\vect{\pi}}{(F+G)^2[1-\probpi{\vect{\pi}}{\vect{Z};\Gamma_f\cap\Gamma_g}]}.
    \label{eq:covariance_bounds}
\end{align}
As the random variables $F\pm G$ are observable, the expectation on either side of \eqref{eq:covariance_bounds} is identifiable, so they provide unbiased estimators of those bounds.
\end{proposition}

\noindent\textbf{Interpretation.} Only units whose treatments enter both functions can induce covariance. If $\Gamma_f\cap\Gamma_g=\emptyset$, both bounds collapse to zero; otherwise, the observable squared sum and difference replace unidentified cross-assignment products, trading sharpness for identification.

Apply the upper bound to $\widehat{\theta}\triangleq n^{-1}\sum_{i=1}^n\widehat{\theta}_i$, where the unit-level estimator $\widehat{\theta}_i$ depends on the treatment assignments of units in $\Gamma_i$. An unbiased estimator of an upper bound $V^+\ge\var{\widehat{\theta}}$ is:
\begin{equation*}
    \widehat{V^+}=\frac{1}{n^2}\left(2\sum_{i=1}^n\widehat{\theta}_i^2(1-\probpi{\vect{\pi}}{\vect{Z};\Gamma_i})+\sum_{i<j}(\widehat{\theta}_i+\widehat{\theta}_j)^2(1-\probpi{\vect{\pi}}{\vect{Z};\Gamma_i\cap\Gamma_j})\right).
\end{equation*}
For the off-policy estimator \eqref{eq:offpolicy_estimator}, $\widehat{\theta}_i$ is its $i^{\text{th}}$ summand before division by $n$, and $\Gamma_i$ contains every unit whose treatment assignment can affect $\widehat{\theta}_i$, whether through the outcome, either exposure indicator, or either exposure-specific weight.

Unbiasedness of $\widehat{V^+}$ for a variance bound does not by itself guarantee finite-sample coverage for a plug-in normal interval. If the Doob CLT in Theorem \ref{theorem:martingaleclt} holds and $\frac{\widehat{V^+}}{V^+}\xrightarrow{p}1$, however, the interval
\begin{equation*}
    \widehat{\theta}\pm z_{1-\alpha/2}\sqrt{\widehat{V^+}}
\end{equation*}
has asymptotic coverage at least $1-\alpha$.

\section{Policy curve and illustrations}\label{sec:eao_curve}

Suppose the design and target policies are homogeneous, so that $\pi_i=p$ and $\pi_i'=q$ for every $i$. Define the common policy shift coefficient:
\begin{equation*}
    \Delta_{pq}\triangleq\frac{q-p}{\sqrt{p(1-p)}},
\end{equation*}
and the unit-averaged order-$k$ Fourier coefficient:
\begin{equation*}
    \mathbb{Y}_k(p)\triangleq\frac{1}{n}\sum_{i=1}^n\sum_{\substack{V\subseteq[n]\\\abs{V}=k}}\hat{y}_{i,p}(V).
\end{equation*}
Fix the design probability $p$ and expand each outcome in the $p$-biased Fourier basis. Taking its expectation under the target probability $q$ and \eqref{eq:target_character_mean} gives the exact policy-curve expansion:
\begin{equation*}
    \eao{q}=\sum_{k=0}^n\Delta_{pq}^{k}\mathbb{Y}_k(p)=\sum_{k=0}^n\frac{(q-p)^k}{k!}\eaoderiv{k}{p},
\end{equation*}
where we use the fact that the coefficients of a centered polynomial in $p$ must encode the corresponding derivatives in the exact Taylor expansion at $p$, i.e. $\eaoderiv{k}{p}$ denotes the $k^{\text{th}}$ derivative of the policy curve at $p$ and:
\begin{equation}
    \label{eq:eao_derivative_fourier_level}
    \eaoderiv{k}{p}=\frac{k!}{(p(1-p))^{\frac{k}{2}}}\cdot\mathbb{Y}_k(p).
\end{equation}
The Fourier level $k$ is not merely analogous to the $k^{\text{th}}$ derivative. After the (known) normalization in \eqref{eq:eao_derivative_fourier_level}, it \emph{is} that derivative. For example, the signed level-$k$ weight:
\begin{equation*}
    \frac{k!}{(p(1-p))^{\frac{k}{2}}} \sum_{\substack{V\subseteq[n]\\\abs{V}=k}}\chi_p^V(\vect{Z}),
\end{equation*}
used appropriately in $n^{-1}\sum_iY_i(\cdot)$ gives an unbiased estimator of $\eaoderiv{k}{p}$. Under the direct- and indirect-effect definitions of \cite{hu2022interference}, the first derivative has the familiar decomposition into the expected average treatment and indirect effects:
\begin{equation*}
    \eaoderiv{1}{p}=\eate{p}+\eite{p},
\end{equation*}
whereas the higher-order derivatives can be seen as encoding higher-order networked causal contrasts. This identity also makes precise what a degree restriction does. Define:
\begin{align*}
    w_{pq}^{\le d}(\vect{z})\triangleq\sum_{\substack{V\subseteq[n]\\\abs{V}\le d}} \Delta_{pq}^{\,|V|}\chi_p^V(\vect{z}),\qquad\eaoest{\le d}{q}&\triangleq\frac{1}{n}\sum_{i=1}^nY_iW_{pq}^{\le d}.
\end{align*}
Then the three views of the same truncation line up:
\begin{equation*}
    \underbrace{\expectpi{p}{\eaoest{\le d}{q}}}_{\text{degree-$d$ transport}}=\underbrace{\sum_{k=0}^d\Delta_{pq}^{\,k}\mathbb{Y}_k(p)}_{\text{retained Fourier levels}}=\underbrace{\sum_{k=0}^d\frac{(q-p)^k}{k!}\eaoderiv{k}{p}}_{\text{Taylor expansion at $p$}}.
\end{equation*}
Consequently, its misspecification bias is exactly the negative Taylor remainder:
\begin{equation*}
    \expectpi{p}{\eaoest{\le d}{q}}-\eao{q}=-\sum_{k=d+1}^n\frac{(q-p)^k}{k!}\eaoderiv{k}{p}.
\end{equation*}
For a fixed finite population there is no separate analyticity assumption, since $\eao{q}$ is a polynomial of degree at most $n$, and the expansion becomes exact at $d=n$. For a sequence of growing populations and truncation orders, estimator bias converges to zero precisely when the policy-weighted omitted Fourier tail converges to zero; Proposition \ref{prop:misspecification_bias} gives a convenient sufficient bound. At fixed $n$, a degree-$d$ truncation has local bias $\order{\abs{q-p}^{d+1}}$ as $q\to p$, while its variance generally increases with $d$. Curvature in the policy curve therefore reflects higher-order interference, but estimating that curvature becomes less precise farther from the experimental policy, as predicted by \eqref{eq:spectral_weight_variance}. The contrast $\eao{1}-\eao{0}$ is the endpoint contrast that is often itself of interest under network interference, called the global average treatment effect (GATE), or total treatment effect (TTE) \cite{savje2021ateunknown}.

Fig. \ref{fig:offpolicy_linnonlin} illustrates this bias--variance trade-off in simulations with linear and nonlinear neighborhood interference in a simple random network. 
\begin{figure}[t!]
\centering
\includegraphics[width=\textwidth]{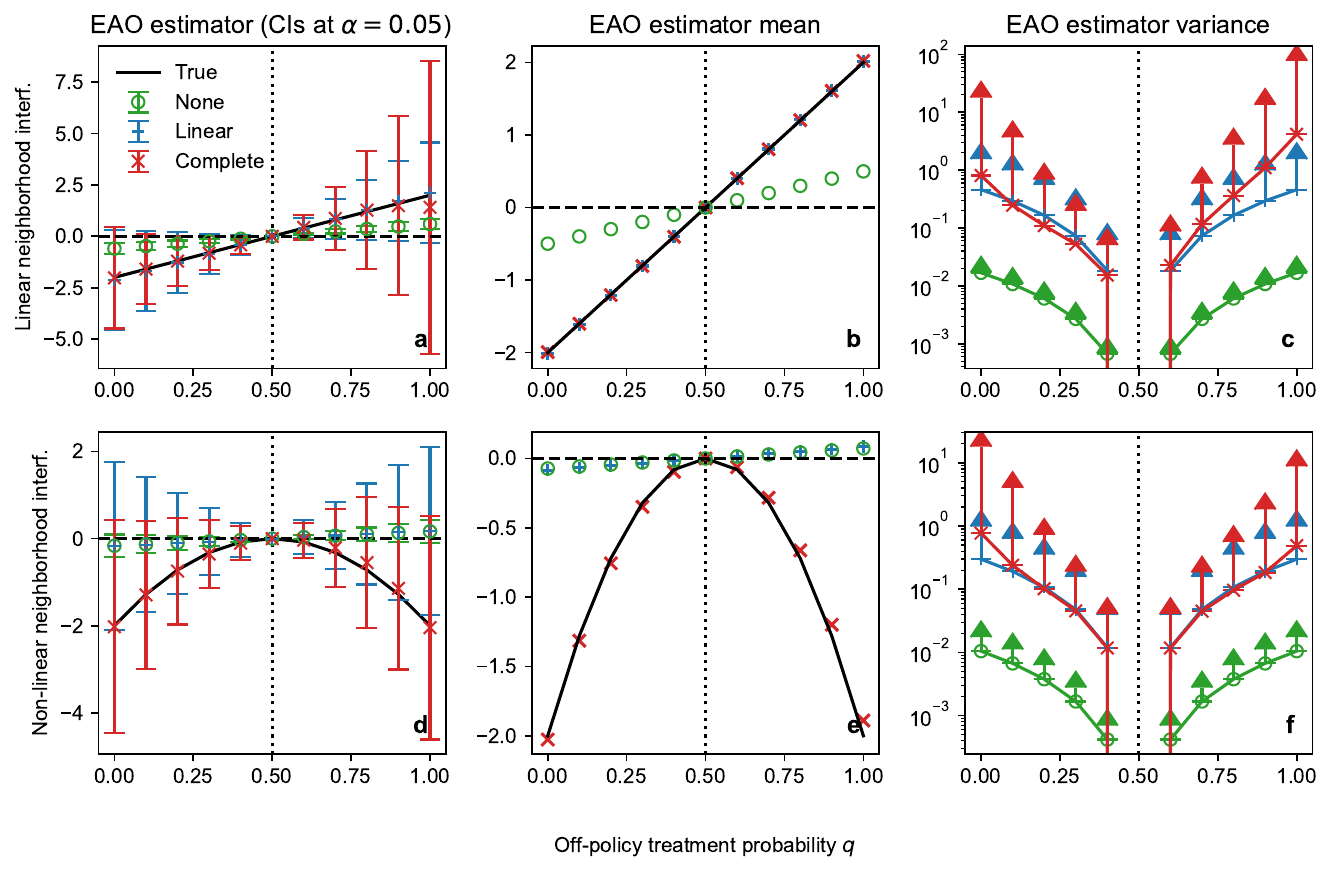}
\caption{\textbf{Simulation under linear and nonlinear neighborhood interference.} A random network with $n=1000$ and mean degree $4$ was generated once and held fixed. Treatments were independently assigned with probability $p=0.5$, and outcomes followed $y_i=\alpha+\bar z_i(\beta+\gamma(\abs{\mathcal{N}_i}\bar z_i-1)/(\abs{\mathcal{N}_i}-1))$, where $\bar z_i=\abs{\mathcal{N}_i}^{-1}\sum_{j\in\mathcal{N}_i}z_j$ and $\mathcal{N}_i$ includes $i$. We used $(\alpha,\beta,\gamma)=(2,4,0)$ in \textbf{a}--\textbf{c} and $(2,8,-8)$ in \textbf{d}--\textbf{f}. The target is the off-policy change $\eao{q}-\eao{p}$. Estimators use the SUTVA ($\circ$), linear ($+$), or complete-local ($\times$) Fourier subspaces in Table \ref{tab:interf_weights}. Panels \textbf{a} and \textbf{d} show one realization with nominal 95\% intervals based on $\widehat{V^+}$; panels \textbf{b} and \textbf{e} show empirical means across 1000 assignments; and panels \textbf{c} and \textbf{f} show empirical variances, with triangles denoting the mean estimated variance bound. The complete-local estimator is unbiased in both settings but becomes unstable farther from $p$; the lower-order estimators trade variance for bias when their Fourier restriction is misspecified.}
\label{fig:offpolicy_linnonlin}
\end{figure}

\section{Discussion}\label{sec:discussion}
Off-policy causal estimation under network interference is possible only to the extent that the interference structure prevents the target policy from demanding information about the entire treatment vector. The spectral formulation in this paper makes this requirement explicit. A chosen Fourier support set determines a transport weight, its minimum variance over the corresponding Fourier subspace, and the rate at which precision deteriorates with Fourier neighborhood size and policy distance. The same formulation nests local inverse-probability and low-order estimators \citep{cortez2022staggered,yu2022estimating}, clarifying the bias--variance trade-off when the assumed interference structure is simplified.

When that simplification is wrong, the resulting bias is not an undifferentiated model error. It is the pairing of the omitted Fourier coefficients of the exposure-weighted outcome with those of the policy shift. Consequently, the same omitted interaction can be harmless for a nearby target policy and important for a more distant one. Therefore, \eqref{eq:exact_misspecification_bias} and \eqref{eq:misspecification_bias_bound} suggest a direct sensitivity analysis: posit bounds on omitted spectral mass and report how the corresponding bias envelope grows along the policy path.

The chi-squared policy distance here is not an arbitrary choice, as it is induced by the $L^2(\mathbb{P}_{\vect{\pi}})$ geometry in which orthogonal projection gives the minimum-variance weight. This geometry suggests a broader calibration program. One could impose the same Fourier balance equations \eqref{eq:spectral_balance_constraints} while minimizing a different convex discrepancy of the target from the design policy. Linear calibration can produce signed weights, whereas entropy and other constrained calibration criteria can favor (or enforce) positive weights \cite{deville1992calibration,hainmueller2012entropy}. We conjecture that such choices yield useful alternative off-policy estimators where exact transport over the balanced subspace would remain but the Fourier closed form and minimum-$L^2$ guarantee would be replaced by the geometry of the chosen objective.

Policy curves can also test spillover mechanisms. Under homogeneous policies, Fourier degree $d$ implies a degree-$d$ polynomial EAO curve. A simultaneous confidence set for the curve, or its coefficients, can therefore be inverted: reject a mechanism class when it contains no curve in the set. Curvature rejects both no interference and additive spillovers, while higher-order restrictions test interaction orders. 

Policy-saturation experiments provide a natural application of this theory. Crépon et al. \cite{crepon2013labormarket} randomized the saturation of job-placement assistance program across labor markets to detect displacement effects on untreated job seekers. In such settings, the slope of the policy curve gives the marginal effect of increasing treatment saturation, while its curvature records how that marginal effect changes as saturation rises. Negative curvature may be consistent with intensifying displacement or congestion, whereas positive curvature may indicate reinforcement or other complementarities. Curvature does not by itself identify a particular spillover mechanism, but it provides evidence of nonlinear interference that can be interpreted alongside substantively motivated decompositions of treated and untreated units.

There is also an experimental design problem beyond estimating one prespecified causal contrast. Network-experiment design work chooses a randomization to control variance for a specified causal estimand \cite{jagadeesan2020design}. Here, however, a single experiment supports a continuum of targets $\set{\eao{\pi}}{\pi\in\Pi}$, with a design-dependent bias--variance envelope at each policy $\pi$. This opens criteria such as integrated risk along a policy path, worst-case risk over a target policy interval, or allocation of precision across select Fourier orders. Optimizing an experiment for an entire policy curve rather than for one point estimand is a worthy open direction.

Finally, we note some limitations of our approach. Correlated designs lose the product basis used in this paper, but two extensions appear plausible. One can construct an orthonormal basis directly under the joint assignment distribution, at the cost of losing the factorization over independent units. Alternatively, when the assignment can be written as $\vect{Z}=\tau(\vect{U})$ for independent randomization seeds $\vect{U}$, one can lift the function to $f\circ\tau$ and apply the product-distribution analysis in the seed space. This extension is more immediate for independent cluster-level assignments, but may also extend to more elaborate randomization designs. The obstacle is that the lift may increase dimension or destroy low-degree sparsity, and a useful theory would need to identify correlated designs admitting a low-complexity lift.

The results also characterize the limits of one realized experiment. As is typical in causal inference in the potential outcomes setup, variance is generally not identified, and here the proposed bounds can be wide. We have stated the central limit theorem for a chosen Doob-martingale reveal. Primitive network-based conditions and constructive filtrations, extensions beyond independent Bernoulli designs, and adaptive Fourier support selection with valid inference remain fruitful directions for future work.

\begin{appendix}

\section{Proofs}\label{sec:proofs}

\subsection*{Proof of Proposition \ref{prop:spectral_transport}}
\begin{proof}
Let $u$ be a candidate weighting function: $u(\vect{z})=\sum_{V\subseteq[n]}\hat{u}_{\vect{\pi}}(V)\chi_{\vect{\pi}}^V(\vect{z})$. For $f\in\mathcal{H}_{\mathcal{S}}$, orthonormality w.r.t. $\vect{\pi}$ yields:
\begin{align*}
    \expectpi{\vect{\pi}}{FU}&=\sum_{V\in\mathcal{S}}\hat{f}_{\vect{\pi}}(V)\,\hat{u}_{\vect{\pi}}(V),\\
    \expectpi{\vect{\pi}'}{F}&=\sum_{V\in\mathcal{S}}\hat{f}_{\vect{\pi}}(V) \Delta_{\vect{\pi}\vect{\pi}'}^V.
\end{align*}
The coefficient vector $\set{\hat{f}_{\vect{\pi}}(V)}{V\in\mathcal{S}}$ can be chosen arbitrarily. In particular, setting $f=\chi_{\vect{\pi}}^V$ shows that equality for the above expressions holds for every $f\in\mathcal{H}_{\mathcal{S}}$ if and only if $\forall V\in\mathcal{S}$:
\begin{equation*}
    \hat{u}_{\vect{\pi}}(V)=\Delta_{\vect{\pi}\vect{\pi}'}^V.
\end{equation*}
Now, the coefficients in $\mathcal{S}^c$ are not restricted by this equality. Parseval's identity gives:
\begin{equation*}
    \normpi{\vect{\pi}}{u}^2=\sum_{V\in\mathcal{S}}\left(\Delta_{\vect{\pi}\vect{\pi}'}^V\right)^2+\sum_{V\notin\mathcal{S}}\hat{u}_{\vect{\pi}}(V)^2.
\end{equation*}
The unique minimum sets every coefficient in the second sum to zero, yielding \eqref{eq:spectral_transport_weight}. Finally, the mean weight $\Delta_{\vect{\pi}\vect{\pi}'}^\emptyset=1$, so subtracting the squared mean of the weight gives \eqref{eq:spectral_weight_variance}.
\end{proof}

\subsection*{Proof of Proposition \ref{prop:misspecification_bias}}
\begin{proof}
The expectation of the weighted function under the design policy is:
\begin{equation*}
    \expectpi{\vect{\pi}}{F W_{\vect{\pi}\vect{\pi}'}^{\mathcal{S}}}=\sum_{V\in\mathcal{S}}\hat{f}_{\vect{\pi}}(V)\,\Delta_{\vect{\pi}\vect{\pi}'}^V.
\end{equation*}
The target expectation is the same sum, except over all subsets. Their difference yields the exact bias in \eqref{eq:exact_misspecification_bias}. Applying Cauchy-Schwarz inequality gives \eqref{eq:misspecification_bias_bound}. To see sharpness, choose the arbitrary omitted coefficient vector for $f$ proportional to $\set{\Delta_{\vect{\pi}\vect{\pi}'}^V}{V\notin\mathcal{S}}$, with norm $c$. This perfectly aligns the two vectors in the Cauchy-Schwarz inequality. If $d_{\mathcal{S}^c}=0$, both sides are zero for every omitted spectrum.
\end{proof}

\subsection*{Proof of Theorem \ref{prop:consistency}}
\begin{proof}
Under the Bernoulli design, $\widehat{\theta}_i^{(n)}$ and $\widehat{\theta}_j^{(n)}$ are independent whenever their influence sets are disjoint. There are at most $n\left(D^{(n)}+1\right)$ ordered unit-pairs with intersecting influence sets (including the units paired with themselves). For every such pair, applying Cauchy-Schwarz inequality yields $\abs{\covpi{\vect{\pi}^{(n)}}{\widehat{\theta}_i^{(n)}}{\widehat{\theta}_j^{(n)}}}\le B^{(n)}$. Summing the covariances and dividing by $n^2$ gives \eqref{eq:consistency_variance_bound}. Applying Chebyshev's inequality then yields consistency.
\end{proof}

\subsection*{Proof of Corollary \ref{prop:spectral_bias_variance}}
\begin{proof}
Proposition \ref{prop:misspecification_bias} gives the exact squared-bias term and bounds its $i^\textsuperscript{th}$ summand by $r_i d_{\mathcal{S}_i^c}(\vect{\pi}',\vect{\pi})$. Moreover,
\begin{equation*}
    \expectpi{\vect{\pi}}{(Y_iW_i)^2}\le M^2\expectpi{\vect{\pi}}{W_i^2}=M^2(1+d_{\mathcal{S}_i}^2(\vect{\pi}',\vect{\pi})).
\end{equation*}
Applying Theorem \ref{prop:consistency} proves the variance term.
\end{proof}

\subsection*{Proof of Proposition \ref{prop:variance_nonidentification}}
\begin{proof}
The expectation of any statistic based on observables $(\vect{Z},F)$ is additively separable over the values in $\set{f(\vect{z})}{\vect{z}\in\booln}$, because only one value $f(\vect{Z})$ is realized. But:
\begin{equation*}
    \variancepi{\vect{\pi}}{F}=\expectpi{\vect{\pi}}{F^2}-\expectpi{\vect{\pi}}{F}^2=\sum_{\vect{z}}\probpi{\vect{\pi}}{\vect{z}}f(\vect{z})^2-\sum_{\vect{z},\vect{z}'}\probpi{\vect{\pi}}{\vect{z}}\probpi{\vect{\pi}}{\vect{z}'}f(\vect{z})f(\vect{z}')
\end{equation*}
contains products of potential outcomes at distinct assignments. Such cross-assignment products cannot be given by an additively separable expectation over unrestricted $f$.
\end{proof}

\subsection*{Proof of Proposition \ref{prop:covariance_bounds}}
\begin{proof}
We use the notation $\vect{Z}_{\Gamma}$ to restrict a treatment vector to the units in $\Gamma$. For each assignment vector $\vect{z}$ on the units in $\Gamma_f\cap\Gamma_g$, let:
\begin{align*}
    p_{\vect{z}}&\triangleq\probpi{\vect{\pi}}{\vect{Z}_{\Gamma_f\cap\Gamma_g}=\vect{z}},\\
    m_f(\vect{z})&=\expectpi{\vect{\pi}}{F\mid\vect{Z}_{\Gamma_f\cap\Gamma_g}=\vect{z}},&m_g(\vect{z})&=\expectpi{\vect{\pi}}{G\mid\vect{Z}_{\Gamma_f\cap\Gamma_g}=\vect{z}}.
\end{align*}
Conditional on $\vect{Z}_{\Gamma_f\cap\Gamma_g}=\vect{z}$, the units outside of $\Gamma_f\cap\Gamma_g$ whose treatments enter $F$ and $G$ are disjoint, and hence don't contribute to their covariance. Thus:
\begin{equation*}
    \covpi{\vect{\pi}}{F}{G}=\sum_{\vect{z}}p_{\vect{z}}(1-p_{\vect{z}})\,m_f(\vect{z})\,m_g(\vect{z})-\sum_{\vect{z}\ne \vect{z}'}p_{\vect{z}}p_{\vect{z}'}\,m_f(\vect{z})\,m_g(\vect{z}').
\end{equation*}
Applying $-uv\le\frac{u^2+v^2}{2}$ to the second sum and collecting by assignment vectors gives:
\begin{equation*}
    \covpi{\vect{\pi}}{F}{G}\le\frac{1}{2}\sum_{\vect{z}}p_{\vect{z}}(1-p_{\vect{z}})(m_f(\vect{z})+m_g(\vect{z}))^2.
\end{equation*}
This bound still depends on squared conditional means, that are not observable from a single assignment---as in the proof for Proposition \ref{prop:variance_nonidentification}. Conditional Jensen's inequality bounds this further by pushing the square inside the expectation, and applying the law of total expectation then yields the upper expression in \eqref{eq:covariance_bounds}. Similarly, applying $-uv\ge-\frac{u^2+v^2}{2}$ gives:
\begin{equation*}
    \covpi{\vect{\pi}}{F}{G}\ge-\frac{1}{2}\sum_{\vect{z}}p_{\vect{z}}(1-p_{\vect{z}})(m_f(\vect{z})-m_g(\vect{z}))^2,
\end{equation*}
and conditional Jensen's inequality followed by the law of total expectation makes this bound identifiable. If $\Gamma_f\cap\Gamma_g=\emptyset$, both bounds equal zero, as they must because $F$ and $G$ are then independent.
\end{proof}

\section{A Doob-martingale CLT}\label{sec:doobclt} 

For each $n$, let $f_n:\booln\to\real$ be square-integrable under a positive Bernoulli design $\vect{\pi}_n\in(0,1)^n$, and let $F_n\triangleq f_n(\vect{Z}_n)$, $\mu_n\triangleq\expectpi{\vect{\pi}_n}{F_n}$, and $\sigma_n^2\triangleq\variancepi{\vect{\pi}_n}{F_n}$. Let $\mathcal{F}_{n,0}\subseteq\cdots\subseteq\mathcal{F}_{n,m_n}$ be any filtration such that $\mathcal{F}_{n,0}$ is trivial and $F_n$ is $\mathcal{F}_{n,m_n}$-measurable. Define the Doob martingale and its differences by
\begin{equation*}
    M_{n,k}\triangleq\condexpectpi{\vect{\pi}_n}{F_n}{\mathcal{F}_{n,k}},\qquad D_{n,k}\triangleq M_{n,k}-M_{n,k-1}.
\end{equation*}
Then $M_{n,0}=\mu_n$, $M_{n,m_n}=F_n$, and $\expectpi{\vect{\pi}_n}{D_{n,k}\mid\mathcal{F}_{n,k-1}}=0$. The unit-wise reveal is $\mathcal{F}_{n,k}=\sigma(Z_{n1},\ldots,Z_{nk})$, where $\sigma(\cdot)$ denotes the sigma-field generated by the first $k$ revealed treatment assignments i.e. all information available after observing them. While this is a natural filtration choice after fixing an ordering of the units \cite{odonnell2014boolean}, it is not intrinsic: the ordering is arbitrary, and valid filtrations may instead reveal blocks or functions of the assignment. 

\begin{theorem}[Doob-martingale CLT]\label{theorem:martingaleclt}
Let $\sigma_n^2>0$ and, as $n\to\infty$,
\begin{align}\label{eq:doob_max_increment}
    \frac{1}{\sigma_n^2}\expectpi{\vect{\pi}_n}{\max_{k\in[m_n]}D_{n,k}^2}&\to 0,\\
    \label{eq:doob_quadratic_variation}
    \frac{1}{\sigma_n^2}\sum_{k=1}^{m_n}D_{n,k}^2&\xrightarrow{p}1.
\end{align}
Then:
\begin{equation*}
    \frac{F_n-\mu_n}{\sigma_n}\xrightarrow{d}\mathcal{N}(0,1).
\end{equation*}
\end{theorem}

\begin{proof}
The normalized differences $\frac{D_{n,k}}{\sigma_n}$ form a martingale-difference array and sum to $\frac{F_n-\mu_n}{\sigma_n}$.  Condition \eqref{eq:doob_max_increment} implies that the largest (normalized) increment converges to zero in probability while also ensuring that the row maxima are uniformly bounded in $L^2$.  Condition \eqref{eq:doob_quadratic_variation} requires their sum of squares---the realized quadratic variation---to converge to one.  Therefore, the result follows from the martingale central limit theorem of McLeish \cite{mcleish1972martingaleclt}.
\end{proof}

\noindent\textbf{Interpretation.} Both CLT conditions are stated in terms of the squared Doob differences $D_{n,k}^2$. However, they rule out two distinct failure modes: \eqref{eq:doob_max_increment} captures a single reveal contributing a non-vanishing share of the total variance, and \eqref{eq:doob_quadratic_variation} captures the accumulated squared differences failing to concentrate around that variance.

\begin{corollary}[Off-policy asymptotic normality]\label{prop:offpolicy_clt}
Let $\widehat{\delta}_n$ be an unbiased off-policy estimator of $\delta_n$ that is square-integrable under $\vect{Z}_n\sim\mathbb{P}_{\vect{\pi}_n}$, and let $\sigma_n^2=\variancepi{\vect{\pi}_n}{\widehat{\delta}_n}$. If its Doob differences under some reveal filtration satisfy \eqref{eq:doob_max_increment}--\eqref{eq:doob_quadratic_variation}, then
\begin{equation*}
    \frac{\widehat{\delta}_n-\delta_n}{\sigma_n}\xrightarrow{d}\mathcal{N}(0,1).
\end{equation*}
\end{corollary}

\subsection*{Conditions under unit-wise reveals}\label{sec:unit_reveals}

The theorem has been deliberately stated at the level of a chosen reveal filtration. We now consider a unit-wise reveal, for which the expansion in the $\vect{\pi}_n$-biased Fourier basis gives some intuition on what the CLT conditions demand:
\begin{equation}
    \label{eq:doob_fourier_difference}
    D_{n,k}=\sum_{\substack{V\subseteq[k]\\\max V=k}}\hat{f}_{n,\vect{\pi}_n}(V)\,\chi_{\vect{\pi}_n}^V(\vect{Z}_n),
\end{equation}
where $\max\emptyset\triangleq 0$. Understanding the conditions of Theorem \ref{theorem:martingaleclt} require interpreting $D_{n,k}^2$, for which it would be helpful to define:
\begin{equation*}
    \lambda_{ni}\triangleq\frac{1-2\pi_{ni}}{\sqrt{\pi_{ni}(1-\pi_{ni})}}, \qquad \lambda_{n,V}\triangleq\prod_{i\in V}\lambda_{ni}.
\end{equation*}
For a single unit $i$: $\left(\chi_{\vect{\pi}_n}^{\{i\}}(\vect{Z}_n)\right)^2=1+\lambda_{ni}\chi_{\vect{\pi}_n}^{\{i\}}(\vect{Z}_n)$, which implies the multiplication rule:
\begin{equation}
    \label{eq:biased_character_product}
    \chi_{\vect{\pi}_n}^{V}(\vect{Z}_n)\,\chi_{\vect{\pi}_n}^{W}(\vect{Z}_n)=\sum_{U\subseteq V\cap W}\lambda_{n,U}\,\chi_{\vect{\pi}_n}^{(V\triangle W)\cup U}(\vect{Z}_n),
\end{equation}
where $\triangle$ denotes the symmetric set difference. Consequently, substituting \eqref{eq:biased_character_product} into \eqref{eq:doob_fourier_difference} gives an exact Fourier expansion of $\sum_kD_{n,k}^2$ under any positive product policy. At the balanced policy $\vect{\pi}=\vect{\frac{1}{2}}$, all nonempty $\lambda_{n,U}$ vanish and \eqref{eq:biased_character_product} reduces to the XOR rule $\chi_{\vect{\frac{1}{2}}}^V(\vect{Z})\,\chi_{\vect{\frac{1}{2}}}^W(\vect{Z})=\chi_{\vect{\frac{1}{2}}}^{V\triangle W}(\vect{Z})$. This yields a compact and interpretable expression for the Fourier expansion of $\sum_k D_{n,k}^2$ from \eqref{eq:doob_fourier_difference}:
\begin{align}
    \sum_{k=1}^{m_n} D_{n,k}^2=\sum_{\substack{V,W\subseteq[m_n],\\\max V>\max W}}\hat{f}_{n,\vect{\frac{1}{2}}}(V)\hat{f}_{n,\vect{\frac{1}{2}}}(V\triangle W)\,\chi_{\vect{\frac{1}{2}}}^W(\vect{Z}_n).
\end{align}
That is, for a balanced design, the quadratic variation captures the reveal-respecting XOR-autocorrelation of the Fourier spectrum of the projection of $f$ onto the revealed units.

While unit-wise reveals are natural for product policies \cite{odonnell2014boolean}, a filtration adapted to blocks or other functions of the assignment could potentially be more informative. Finding primitive network or Fourier conditions that imply \eqref{eq:doob_max_increment}--\eqref{eq:doob_quadratic_variation} for other useful filtrations remains a separate open question. 

\end{appendix}

\bibliography{bibliography}

\end{document}